\documentclass[aps,a4paper,nofootinbib,notitlepage,twocolumn]{revtex4-1}
\usepackage{hyperref,graphicx, amsmath, amssymb, amsthm, color, bm, 
bbm,dsfont,cleveref}

\usepackage[shortlabels]{enumitem}

\theoremstyle{plain}
\newtheorem{thm}{Theorem}

\theoremstyle{definition}

\theoremstyle{remark}

\definecolor{nblue}{rgb}{0.2,0.2,0.7}
\definecolor{ngreen}{rgb}{0.1,0.5,0.1}
\definecolor{nred}{rgb}{0.8,0.2,0.2}
\definecolor{nblack}{rgb}{0,0,0}

\providecommand{\bra}[1]{\langle #1 |}
\providecommand{\ket}[1]{|#1 \rangle}

\newcommand{\md}[1]{\mathds{#1}}
\newcommand{\mc}[1]{\mathcal{#1}}

\newcommand{\mbf}[1]{\mathbf{#1}}
\renewcommand{\>}{\rangle}
\newcommand{\<}{\langle}

\newcommand{\tr}{\mathrm{Tr}}

\newcommand{\tn}[1]{^{\otimes #1}}
\newcommand{\ct}{^{\dagger}}

\newcommand{\hidden}[1]{}

\crefname{equation}{equation}{equations}
\crefname{figure}{figure}{figures}
\crefname{thm}{theorem}{theorems}
\begin{document}

\title{Noise tailoring for scalable quantum computation via randomized 
compiling} 

\author{Joel J. \surname{Wallman}}
\affiliation{Institute for Quantum Computing and Department of Applied 
Mathematics, University of Waterloo, Waterloo, Canada}

\author{Joseph \surname{Emerson}}
\affiliation{Institute for Quantum Computing and Department of Applied 
Mathematics, University of Waterloo, Waterloo, Canada}
\affiliation{Canadian Institute for Advanced Research, Toronto, Ontario M5G 
1Z8, Canada}

\date{\today}

\begin{abstract}
Quantum computers are poised to radically outperform their classical 
counterparts by manipulating coherent quantum systems. A realistic quantum 
computer will experience errors due to the environment and imperfect 
control. When these errors are even partially coherent, they present a major 
obstacle to performing robust computations. Here, we propose a method for 
introducing independent random single-qubit gates into the logical circuit in 
such a way that the effective logical circuit remains unchanged. We prove that 
this randomization tailors the noise into stochastic Pauli errors, which can 
dramatically reduce error rates while introducing little or no experimental 
overhead. Moreover, we prove that our technique is robust to the inevitable 
variation in errors over the randomizing gates and numerically illustrate the 
dramatic reductions in worst-case error that are achievable. Given such 
tailored noise, gates with significantly lower fidelity---comparable to 
fidelities realized in current experiments---are sufficient to achieve 
fault-tolerant quantum computation. Furthermore, the worst-case error rate of 
the tailored noise can be directly and efficiently measured through randomized 
benchmarking protocols, enabling a rigorous certification of the performance of 
a quantum computer. 
\end{abstract}

\maketitle

The rich complexity of quantum states and processes enables powerful new 
protocols for processing and communicating quantum information, as illustrated 
by Shor's factoring algorithm~\cite{Shor1999} and quantum simulation 
algorithms~\cite{Lloyd1996}. However, the same rich complexity of quantum 
processes that makes them useful also allows a large variety of errors to 
occur. Errors in a quantum computer arise from a variety of sources, including 
decoherence and imperfect control, where the latter generally lead to coherent (unitary) errors. 
It is provably possible to perform a fault-tolerant quantum computation in the 
presence of such errors provided they occur with at most some maximum threshold 
probablity~\cite{Shor1995,Gottesman1996,Knill1998,Aharonov1999,Calderbank2002,Kitaev2003}.
However, the fault-tolerant threshold probability depends upon the error-correcting code and is 
notoriously difficult to estimate or bound because of the sheer variety of possible 
errors. Rigorous lower bounds on the threshold of the order of $10^{-6}$~\cite{Aharonov1999} 
for generic local noise and $10^{-4}$~\cite{Aliferis2007a} and 
$10^{-3}$~\cite{Aliferis2008} for stochastic Pauli noise have been obtained for 
a variety of codes. While these bounds are rigorous, they are far below 
numerical estimates that range from $10^{-2}$~\cite{Knill2005,Wang2011} and 
$10^{-1}$~\cite{Duclos-Cianci2010,Wootton2012,Bombin2012}, which are generally 
obtained assuming stochastic Pauli noise, largely because the 
effect of other errors is too difficult to simulate~\cite{Puzzuoli2014}. While 
a threshold for Pauli errors implies a threshold exists for arbitrary errors 
(e.g., unitary errors), there is currently no known way to rigorously estimate 
a threshold for general noise from a threshold for Pauli noise. 

The ``error rate'' due to an arbitrary noise map $\mc{E}$ can be quantified in 
a variety of ways. Two particularly important quantities are the average error 
rate defined via the gate fidelity
\begin{align}
r(\mc{E}) = 1- \int {\rm d}\psi 
\bra{\psi}\mc{E}(\ket{\psi}\!\bra{\psi})\ket{\psi}
\end{align}
and the worst-case error rate (also known as the diamond distance from the 
identity)~\cite{Kitaev1997}
\begin{align}\label{eq:diamond_def}
\epsilon(\mc{E}) 
= \tfrac{1}{2}\| \mc{E} - \mc{I}\|_{\diamond} 
= \sup_{\psi} \tfrac{1}{2}\| 
\left(\mc{E}\otimes\mc{I}_d-\mc{I}_{d^2}\right)(\psi)\|_1
\end{align}
where $d$ is the dimension of the system $\mc{E}$ acts on, $\|A\|_1 = 
\sqrt{\tr A\ct A}$ and the maximization is over all $d^2$-dimensional pure 
states (to account for the error introduced when acting on entangled states). 
The average error rate $r(\mc{E})$ is an experimentally-convenient 
characterization of the error rate because it can be efficiently estimated via 
randomized benchmarking~\cite{Emerson2005, Emerson2007, Dankert2009, Knill2008, 
Magesan2011}. However, the diamond distance is typically the quantity used to 
prove rigorous fault-tolerance thresholds~\cite{Aharonov1999}. The average 
error rate and the worst-case error rate are related via the 
bounds~\cite{Beigi2011, Wallman2014}
\begin{align}\label{eq:fidelity_to_worst}
r(\mc{E}) d^{-1}(d+1)
\leq \epsilon(\mc{E}) 
\leq \sqrt{r(\mc{E})} \sqrt{d(d+1)}.
\end{align}
The lower bound is saturated by any stochastic Pauli noise, in which case the 
worst-case error rate is effectively equivalent to the experimental estimates 
obtained efficiently via randomized benchmarking~\cite{Magesan2012a}. While the 
upper bound is not known to be tight, there exist unitary channels such that 
$\epsilon(\mc{E}) \approx \sqrt{(d+1)r(\mc{E})/4}$, so the scaling with $r$ 
is optimal~\cite{Sanders2015}. 

The scaling of the upper bound of \cref{eq:fidelity_to_worst} is only saturated 
by purely unitary noise. However, even a small coherent error relative to 
stochastic errors can result in a dramatic increase in the worst-case error.  
For example, consider a single qubit noise channel with $ r = 1\times 10^{-4}$, 
where the contribution due to stochastic noise processes is $r = 0.83\times 
10^{-4}$ and the remaining contribution is from a small unitary (coherent) 
rotation error. The worst-case error for such noise is $\epsilon \approx 
10^{-2}$, essentially two orders of magnitude greater than the 
infidelity~\cite{Sanders2015}.

Here we show that by compiling random single-qubit gates into a logical 
circuit, noise with arbitrary coherence and spatial correlations can be 
tailored into stochastic Pauli noise. We also prove that our technique is 
robust to gate-dependent errors which arise naturally due to imperfect 
gate calibration. In particular, our protocol is fully robust against arbitrary 
gate-dependent errors on the gates that are most difficult to implement, while 
imperfections in the easier gates introduces an additional error that is 
essentially proportional to the infidelity.

Our randomized compiling technique requires only a small (classical) overhead 
in the compilation cost, or, alternatively,  can be implemented on the fly 
with fast classical control. Stochastic Pauli errors with the same average 
error rate $r$ as a coherent error leads to four major advantages for quantum 
computation: (i) they have a substantially lower worst-case error rate, (ii) 
the worst-case error rate can be directly estimated efficiently and robustly 
via randomized benchmarking experiments, enabling a direct comparison to a 
threshold estimate to determine if fault-tolerant quantum computation is 
possible, (iii) the known fault-tolerant thresholds for Pauli errors are 
substantially higher than for coherent errors, and (iv) the average error rate 
accumulates linearly with the length of a computation for stochastic Pauli 
errors, whereas it can accumulate quadratically for coherent errors.

Randomizing quantum circuits has been previously proposed in 
Refs~\cite{Knill2004a,Kern2005}. However, these proposals have specific 
limitations that our technique circumvents. The proposal for inserting Pauli 
gates before and after Clifford gates proposed in Ref.~\cite{Knill2004a} is a 
special case of our technique when the only gates to be implemented are 
Clifford gates. However, this technique does not account for non-Clifford gates 
whereas our generalized technique does. As a large number of non-Clifford gates 
are required to perform useful quantum computations~\cite{Aaronson2004} and are 
often more difficult to perform fault-tolerantly, our generalized technique 
should be extremely valuable in practice. Moreover, the proposal in 
Ref.~\cite{Knill2004a} assumes that the Pauli gates are essentially perfect, 
whereas we prove that our technique is robust to imperfections in the Pauli 
gates. Alternatively, Pauli-Random-Error-Correction (PAREC) has been shown to 
eliminate static coherent errors~\cite{Kern2005}. However, PAREC involves 
changing the multi-qubit gates in each step of the computation. As multi-qubit 
errors are currently the dominant error source in most experimental platforms 
and typically depend strongly on the gate to be performed, it is unclear how 
robust PAREC will be against gate-dependent errors on multi-qubit gates and 
consequently against realistic noise. By way of contrast, our technique is 
completely robust against arbitrary gate-dependent errors on multi-qubit gates.

\section{Results}

\subsection{Standardized form for compiled quantum circuits}

We begin by proposing a standardized form for compiled quantum circuits based 
on an operational distinction between `easy' and `hard' gates, that 
is, gates that can be implemented in a given experimental platform with 
relatively small and large amounts of noise respectively. We also propose a 
specific choice of easy and `hard' gates that is well-suited to many 
architectures for fault-tolerant quantum computation.

In order to experimentally implement a quantum algorithm, a quantum circuit is 
compiled into a sequence of elementary gates that can be directly implemented 
or have been specifically optimized. Typically, these elementary gates can be 
divided into easy and hard gate sets based either on how many physical qubits 
they act on or how they are implemented within a fault-tolerant architecture. 
In near-term applications of universal quantum computation without quantum 
error correction, such as quantum simulation, the physical error model and 
error rate associated with multi-qubit gates will generally be distinct from, 
and much worse than, those associated with single qubit gates. In the 
long-term, fault-tolerant quantum computers will implement some operations 
either transversally (that is, by applying independent operations to a set of 
physical qubits) or locally in order to prevent errors from cascading. However, 
recent `no-go' theorems establish that for any fault-tolerant scheme, there 
exist some operations cannot be performed in such a 
manner~\cite{Eastin2009,Beverland2014} and so must be implemented via 
other means, such as magic-state injection~\cite{Bravyi2005} or gauge 
fixing~\cite{Bombin2015}. 

The canonical division that we consider is to set the easy gates to be the 
group generated by Pauli gates and the phase gate $R = |0\>\!\< 0| + i|1\>\!\< 
1|$, and the hard gate set to be the Hardamard gate $H$, the $\pi/8$ gate 
$\sqrt{R}$ and the two-qubit controlled-Z gate $\Delta(Z) = |0\>\!\<0|\otimes I 
+ |0\>\!\<0|\otimes Z$. Such circuits are universal for quantum computation and 
naturally suit many fault-tolerant settings, including CSS codes with a 
transversal $T$ gate (such as the 15-qubit Reed-Muller code), color codes and 
the surface code. While some of the `hard' gates may be easier than others in a 
given implementation, it is beneficial to make the set of easy gates as small 
as possible since our scheme is completely robust to arbitrary variations in 
errors over the hard gates.

With such a division of the gates, the target circuit can be reorganized into a 
circuit (the `bare' circuit) consisting of $K$ clock cycles, wherein each cycle 
consists of a round of easy gates followed by a round of hard gates applied 
to disjoint qubits as in \cref{fig:randomized_compilation}~a. To 
concisely represent the composite operations performed in individual rounds, we 
use the notational shorthand $\vec{A} = A_1\otimes \ldots \otimes A_n$ and 
define $G_k$ to be the product of all the hard gates applied in the $k$th 
cycle. We also set $G_K=I$ without loss of generality, so that the circuit 
begins and ends with a round of easy gates. 

\begin{figure}	
	\includegraphics[width=0.95\linewidth]{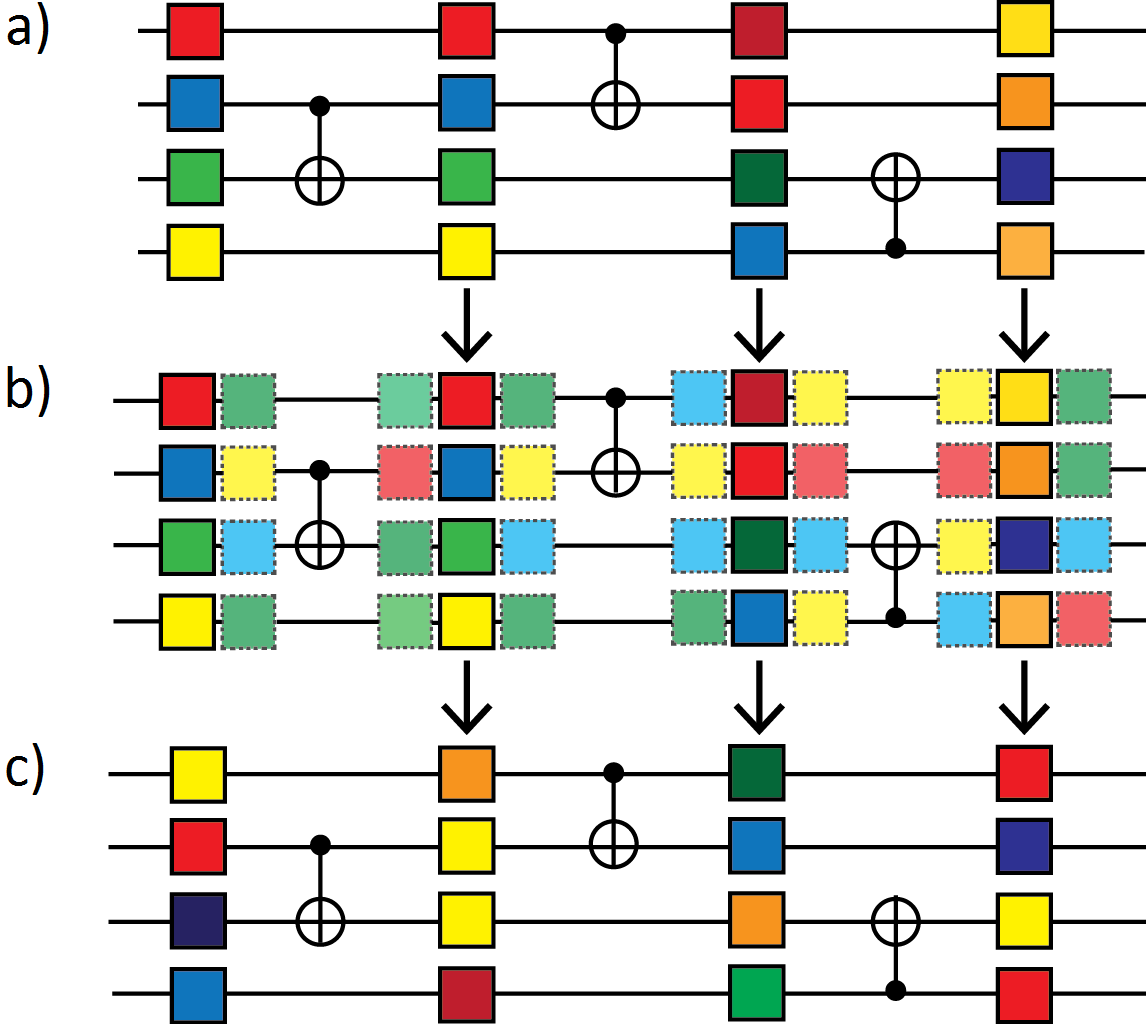}
\caption{a) Example of a bare circuit that is arranged into cycles wherein each 
cycle consists of a round of easy single-qubit gates and a round of hard gates 
(here, the hard gates are controlled-NOT gates). b) A randomized circuit 
wherein twirling gates have been inserted before and after every easy gate. c) 
A randomized circuit wherein the twirling gates have been compiled into the 
easy gates, resulting in a new circuit that is logically equivalent to the bare 
circuit and has the same number of elementary 
gates.}\label{fig:randomized_compilation}
\end{figure}

\subsection{Randomized compiling}

We now specify how standardized circuits in the above form can be randomized in 
order to average errors in the implementations of the elementary gates into an 
effective stochastic channel, that is, into a channel $\mc{E}$ that maps any 
$n$-qudit state $\rho$ to
\begin{align}
\mc{E}(\rho) = \sum_{P\in \mbf{P}_d\tn{n}} c_P P\rho P\ct,
\end{align}
where $\mbf{P}_d\tn{n}$ is the set of $d^{2n}$ generalized Pauli operators and 
the coefficients $c_P$ are a probability distribution over $\mbf{P}_d\tn{n}$. 
For qubits ($d=2$), $\mbf{P}_2$ is the familiar set of four Hermitian and 
unitary Pauli operators $\{I,X,Y,Z\}$.

Let $\mbf{C}$ denote the group generated by the easy gates and assume that it 
contains a subset $\mbf{T}$ such that
\begin{align}\label{eq:reasy}
\mc{E}^{\mbf{T}} = \md{E}_T T\ct \mc{E} T
\end{align}
is a stochastic channel for any channel $\mc{E}$, where $\md{E}_x f(x) = 
|X|^{-1}\sum_{x\in X} f(x)$ denotes the uniform average over a set $X$ 
(typically a gate set implicit from the context). The canonical example of such 
a set is $\mbf{P}_d\tn{n}$ or any group containing $\mbf{P}_d\tn{n}$. 

We propose the following randomized compiling technique, where the 
randomization should ideally be performed independently for each run of a given 
bare circuit. Each round of easy gates $\vec{C}_k$ in the bare circuit of 
\cref{fig:randomized_compilation}~a is replaced with a round of randomized 
dressed gates
\begin{align}\label{eq:dressed}
\tilde{C}_k = \vec{T}_k \vec{C}_k \vec{T}_{k-1}^c
\end{align}
as in \cref{fig:randomized_compilation}~b, where the $T_{j,k}$ are chosen 
uniformly at random from the twirling set $\mbf{T}$ and the correction 
operators are set to $\vec{T}_k^c = G_k\vec{T}_k\ct G_k\ct$ to undo the 
randomization from the previous round. The edge terms $\vec{T}^c_0$ 
and $\vec{T}_K$ can either be set to the identity or also randomized depending 
on the choice of the twirling set and the states and measurements.

The dressed gates should then be compiled into and implemented as a single 
round of elementary gates as in \cref{fig:randomized_compilation}~c rather than 
being implemented as three separate rounds of elementary gates. In order to 
allow the dressed gates to be compiled into a single easy gate, we require 
$\vec{T}_k^c\in\mbf{C}\tn{n}$ for all $\vec{T}_k\in\mbf{T}\tn{n}$. The 
example with $\mbf{T} = \mbf{P}_d$ that has been implicitly appealed to and 
described as ``toggling the Pauli frame'' previously~\cite{Knill2005} is 
a special case of the above technique when the hard gates are Clifford gates 
(which are defined to be the gates that map Pauli operators to Pauli operators 
under conjugation), but breaks down when the hard gates include non-Clifford 
gates such as the single-qubit $\pi/8$ gate. For the canonical division into 
easy and hard gates from the previous section, we set $\mbf{T} = \mbf{P}_2$, 
$\mbf{C}$ to be the group generated by $R$ and $\mbf{P}_2$ (which is isomorphic 
to the dihedral group of order 8) and the hard gates to be rounds of $H$, 
$\sqrt{R}$ and $\Delta(Z)$ gates. Conjugating a Pauli gate by $H$ or 
$\Delta(Z)$ maps it to another Pauli gate, while conjugating by $\sqrt{R}$ maps 
$X^x Z^z$ to $R^x X^x Z^z$ (up to a global phase). Therefore the correction 
gates, and hence the dressed gates, are all elements of the easy gate set.

The tailored noise is not realized in any individual choice of sequences. 
Rather, it is the average over independent random sequences. However, while 
each term in the tailored noise can have a different effect on an input state 
$\rho$, if the twirling gates are independently chosen on each run, then 
the expected noise over multiple runs is exactly the tailored noise. 
Independently sampling the twirling gates each time the circuit is run 
introduces some additional overhead, since the dressed gates (which are 
physically implemented) depend on the twirling gates and so need to be 
recompiled for each experimental run of a logical circuit. However, this 
recompilation can be performed in advance efficiently on a classical computer 
or else applied efficiently `on the fly' with fast classical control. Moreover, 
this fast classical control is exactly equivalent to the control required in 
quantum error correction so imposes no additional experimental burden.

We will prove below that our technique tailors noise satisfying various 
technical assumptions into stochastic Pauli noise. We expect the technique will 
also tailor more general noise into approximately stochastic noise, 
though leave a fully general proof as an open problem. 

\subsection{Robustness to arbitrary independent errors on the hard gates}

We now prove that our randomized compiling scheme results in an average 
stochastic noise channel for Markovian noise that depends arbitrarily upon the 
hard gates but is independent of the easy gate. Under this assumption, the 
noisy implementations of the $k$th round of easy gates $\vec{C}_k$ and hard 
gates $G_k$ can be written as $\mc{E}_{\rm e} \vec{C}$ and $G_k\mc{E}(G_k)$ 
respectively, where $\mc{E}_{\rm e}$ and $\mc{E}(G_k)$ are $n$-qubit channels 
that can include cross-talk and multi-qubit correlations and $\mc{E}(*)$ can 
depend arbitrarily on the argument, that is, on which hard gates are 
implemented.

\begin{thm}\label{thm:gate_independent}
Randomly sampling the twirling gates $\vec{T}_k$ independently in each round 
tailors the noise at each time step (except the last) into stochastic Pauli 
noise when the noise on the easy gates is gate-independent.
\end{thm}

\begin{proof}
The key observation is that if the noise in rounds of easy gates is some fixed 
noise channel $\mc{E}_{\rm e}$, then the dressed gates in 
\cref{eq:dressed} have the same noise as the bare gates and so 
compiling in the extra twirling gates in \cref{fig:noise_compilation}~c 
does not change the noise at each time step, as illustrated in 
\cref{fig:noise_compilation}~a and d. Furthermore, the correction gates 
$T^c_{k,j}$ are chosen so that they are the inverse of the randomizing gates 
when they are commuted through the hard gates, as illustrated in 
\cref{fig:noise_compilation}~b and c. Consequently, uniformly averaging 
over the twirling gates in every cycle reduces the noise in the $k$th cycle to 
the tailored noise
\begin{align}\label{eq:rhard}
\mc{T}_k =\md{E}_{\vec{T}} \vec{T}\ct\mc{E}(G_k) \mc{E}_{\rm e}\vec{T}
\end{align}
where for channels $\mc{A}$ and $\mc{B}$, $\mc{A}\mc{B}$ denotes the channel 
whose action on a matrix $M$ is $\mc{A}(\mc{B}(M))$. When $\mbf{T}=\mbf{P}$, 
the above channel is a Pauli channel~\cite{Emerson2007}. Moreover, by the 
definition of a unitary 1-design~\cite{Dankert2009}, the above sum is 
independent of the choice of $\mbf{T}$ and so is a Pauli channel for any 
unitary 1-design.
\end{proof}

\Cref{thm:gate_independent} establishes that the noise in all but the final 
cycle can be 
exactly tailored into stochastic noise (albeit under somewhat idealized 
conditions which will be relaxed below). To account for noise in the final 
round, we can write the effect corresponding to a measurement outcome 
$|\vec{z}\>$ as $\mc{A}(|\vec{z}\>\!\<\vec{z}|)$ for some fixed noise 
map $\mc{A}$. If $\mbf{P}\subset\mbf{C}$, we can choose $\vec{T}_K$ uniformly 
at random from $\mbf{P}\tn{n}$. A virtual Pauli gate can then be inserted 
between the noise map $\mc{A}$ and the idealized measurement effect 
$|\vec{z}\>\<\vec{z}|$ by classically relabeling the measurement outcomes to 
map $\vec{z}\to \vec{z}\oplus \vec{x}$, where $\oplus$ denotes entry-wise 
addition modulo two. Averaging over $\vec{T}_K$ with this relabeling reduces 
the noise in the final round of single-qubit Clifford gates and the measurement 
to 
\begin{align}
\overline{\mc{A}} &= \md{E}_{\vec{P}} \vec{P} \mc{A} \mc{E}_e \vec{P}
\end{align}
This technique can also be applied to quantum non-demolition measurements on a 
subset of the qubits (as in, for example, error-correcting circuits), where the 
unmeasured qubits have randomizing twirling gates applied before and after the 
measurement.

\begin{figure}	
	\includegraphics[width=0.95\linewidth]{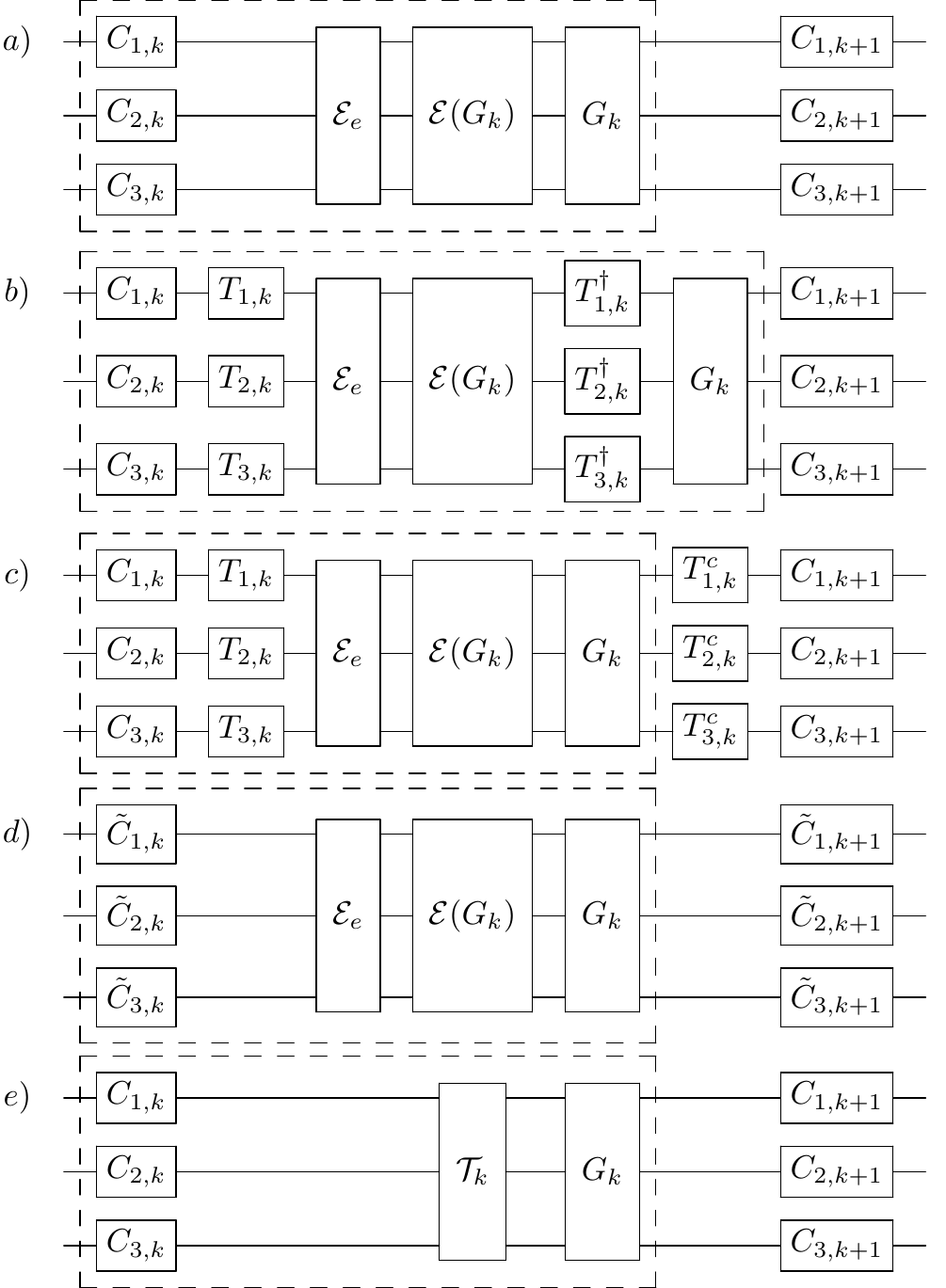}
	\caption{a) Fragment of a noisy bare circuit with the $k$th cycle indicated 
		by the dashed box. b) To tailor the noise into stochastic noise, we 
		insert random twirling gates before the noise and the corresponding 	
		correction gates immediately after the noise. c) Equivalent circuit 
		to b, where the correction gates have been commuted through $G_k$, the 
		round of hard gates. d) The randomized circuit equivalent to b and c, 
		where the twirling and correction gates have been compiled into the 
		adjacent easy gates. e) The tailored circuit obtained by averaging over 
		randomized circuits where $\mc{T}_k$ is the tailored stochastic Pauli 
		channel in \cref{eq:rhard}.}\label{fig:noise_compilation}
\end{figure}

\subsection{Robustness to independent errors on the easy gates}

By \cref{thm:gate_independent}, our technique is fully robust to the most 
important form of gate-dependence, namely, gate-dependent errors on the 
hard gates. However, \cref{thm:gate_independent} still requires that the 
noise on the easy gates is effectively gate-independent. Because residual 
control errors in the easy gates will generally produce small gate-dependent 
(coherent) errors, we will show that the benefits of noise tailoring can still 
be achieved in this physically realistic setting. 

When the noise depends on the easy gates, 
the tailored noise in the $k$th cycle from equation~\eqref{eq:rhard} becomes
\begin{align}
	\mc{T}_k^{\rm GD} =\md{E}_{\vec{T}_1,\ldots,\vec{T}_k} 
	\vec{T}_k\ct\mc{E}(G_k) \mc{E}(\vec{\tilde{C}}_k)\vec{T}_k,
\end{align}
which depends on the previous twirling gates through $\vec{\tilde{C}}_k$ by 
\cref{eq:dressed}. This dependence means that we cannot assign independent 
noise to each cycle in the tailored circuit. 

However, in \cref{thm:gate_dependent} (proven in `Methods') we show 
that implementing a circuit with gate-dependent noise 
$\mc{E}(\vec{\tilde{C}}_k)$ instead of the corresponding gate-independent noise
\begin{align}
\mc{E}_k^{\mbf{T}} = \md{E}_{\vec{\tilde{C}}_k} \mc{E}(\vec{\tilde{C}}_k)
\end{align}
introduces a relatively small additional error. We show that the additional 
error is especially small when $\mbf{T}$ is a group normalized by $\mbf{C}$, 
that is,  $CTC\ct\in\mbf{T}$ for all $C\in\mbf{C}$, $T\in\mbf{T}$. This 
condition is satisfied in many practical cases, including the scenario where 
$\mbf{T}$ is the Pauli group and $\mbf{C}$ is the group generated by Pauli and 
$R$ gates. The stronger bound reduces the contributions from every cycle by 
orders of magnitude in parameter regimes of interest (i.e., 
$\epsilon\left[\mc{E}(G_{k-1})\mc{E}_{k-1}^{\mbf{T}}\right] \leq 10^{-2}$, 
comparable to current experiments), so that the bound on the additional error 
grows very slowly with the circuit length.

\begin{thm}\label{thm:gate_dependent}
Let $\mc{C}_{\rm GD}$ and $\mc{C}_{\rm GI}$ be tailored circuits with 
gate-dependent and gate-independent noise on the easy gates respectively. Then
\begin{align}
\|\mc{C}_{\rm GD}-\mc{C}_{\rm GI}\|_\diamond 
\leq \sum_{k=1}^K \md{E}_{\vec{T}_1,\ldots,\vec{T}_K}  
\|\mc{E}(\vec{\tilde{C}}_k)-\mc{E}_k^{\mbf{T}}\|_\diamond.
\end{align}
When $\mbf{T}$ is a group normalized by $\mbf{C}$, this can be improved to
\begin{align}
\|\mc{C}_{\rm GD},\mc{C}_{\rm GI}\|_\diamond 
&\leq \sum_{k=2}^K 2\md{E}_{\vec{\tilde{C}}_k}\| \mc{E}(\vec{\tilde{C}}_k) - 
\mc{E}_k^{\mbf{T}}\|\epsilon\left[\mc{E}(G_{k-1})\mc{E}_{k-1}^{\mbf{T}}\right]  
\notag\\
&\quad + \md{E}_{\vec{\tilde{C}}_1}\|\mc{E}(\vec{\tilde{C}}_1) - 
\mc{E}_1^{\mbf{T}}\|_\diamond.
\end{align}
\end{thm}

There are two particularly important scenarios in which the effect of 
gate-dependent contributions need to be considered and which determine the physically relevant value of $K$. In near-term applications 
such as quantum simulators, the following theorem would be applied to the 
entire circuit, while in long-term applications with quantum error correction, 
the following theorem would be applied to fragments corresponding to rounds of 
error correction. Hence under our randomized compiling technique, the noise on 
the easy gates imposes a limit either on the useful length of a circuit without 
error correction or on the distance between rounds of error correction. It is 
important to note that a practical limit on $K$ is already imposed, in the 
absence of our technique, by the simple fact that even Pauli noise accumulates 
linearly in time, so $r(\mc{T}_k)\ll 1/K$ is already required to ensure that 
the output of any realistic circuit remains close to the ideal circuit.

While \cref{thm:gate_dependent} provides a very promising bound, it is unclear 
how to estimate the quantities $\tfrac{1}{2}\md{E}_{\vec{\tilde{C}}_k}\| 
\mc{E}(\vec{\tilde{C}}_k) - \mc{E}_k^{\mbf{T}} \|_\diamond$ without performing 
full process tomography. To remedy this, we now provide the following bound in 
terms of the infidelity, which can be efficiently estimated via randomized 
benchmarking. We expect the following bound is not tight as we use the triangle 
inequality to turn the deviation from the average noise into deviations from no 
noise, which should be substantially larger. However, even the following loose 
bound is sufficient to rigorously guarantee that our technique significantly 
reduces the worst-case error, as illustrated in \cref{fig:noise_reduction} for 
a two-qubit gate in the bulk of a circuit (i.e., with $k>1$).

The following bound could also be substantially improved if the noise on the 
easy gates is known to be close to depolarizing (even if the hard gates 
have strongly coherent errors), as quantified by the 
unitarity~\cite{Wallman2015,Kueng2015,Wallman2015b}. However, rigorously 
determining an improved bound would require analyzing the protocol for 
estimating the unitarity under gate-dependent noise, which is currently 
an open problem.

\begin{thm}\label{thm:local}
For arbitrary noise,
\begin{align}
\md{E}_{\vec{\tilde{C}}_k}\| \mc{E}(\vec{\tilde{C}}_k) - 
\mc{E}_k^{\mbf{T}} \|_\diamond 
\leq 2\epsilon(\mc{E}_k^{\mbf{T}}) + 2\sqrt{\md{E}_{\vec{\tilde{C}}_k} 
\epsilon[\mc{E}(\vec{\tilde{C}}_k)]^2}.
\end{align}	
For $n$-qubit circuits with local noise on the easy gates, 
\begin{align}
\md{E}_{\vec{\tilde{C}}_k}\| \mc{E}(\vec{\tilde{C}}_k) - 
\mc{E}_k^{\mbf{T}} \|_\diamond &\leq 
\sum_{j=1}^n 4\sqrt{6r\Bigl[\overline{\mc{E}}_j^{\mbox{T}}(C_{j,k})\Bigr]}
\end{align}
for $k=1,\ldots,K$, where $\mc{E}_{j.k}^{\mbox{T}} 
=\md{E}_{\tilde{C}_{j,k}}\mc{E}_j (\tilde{C}_{j,k})$ is the local noise on the 
$j$th qubit averaged over the dressed gates in the $k$th cycle.
\end{thm}

\begin{figure}
	\includegraphics[height=6cm]{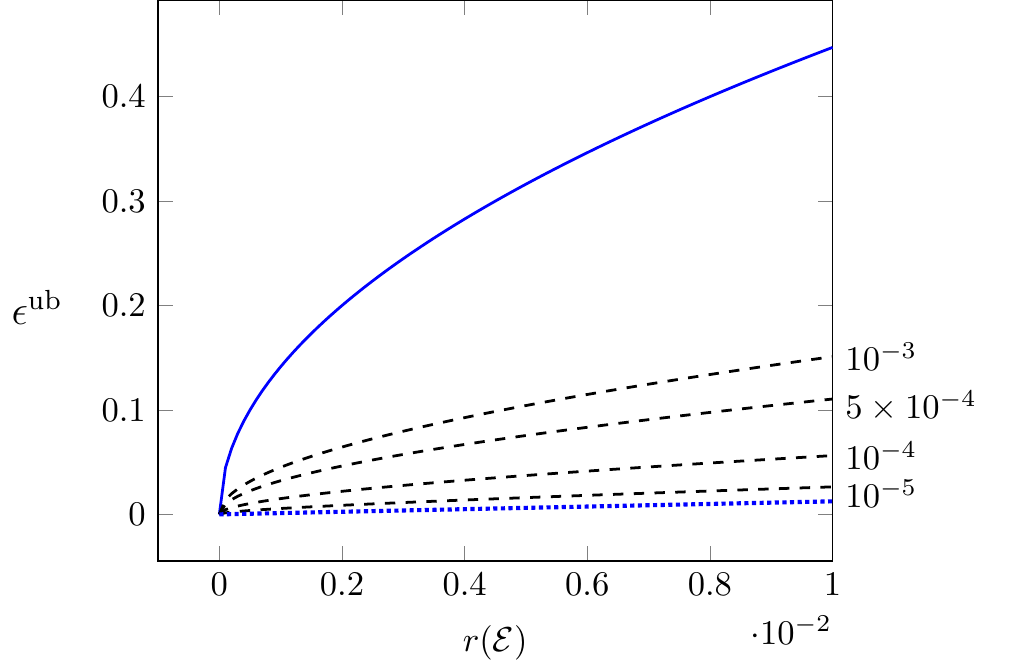}	
\caption{Upper bounds $\epsilon^{\rm ub}$ on the worst case error for a 
two-qubit hard gate in the bulk of a circuit (e.g., a controlled-NOT gate with 
$k>1$) as a function of its infidelity $r(\mc{E})$ with [dashed black, from 
\cref{thm:gate_dependent,thm:local}] and without [solid blue, from 
\cref{eq:fidelity_to_worst}] our tailoring technique under gate-dependent local 
noise on the single-qubit gates with infidelity 
$r(\overline{\mc{E}}^{\mbf{T}}_{j,k})=10^{-5},10^{-4},5\times 			
10^{-3},10^{-3}$ respectively. The worst-case error rate achieved by our 
technique for gate-independent noise (over the dressed gates) is plotted 
for comparison [dotted blue, from 	
\cref{eq:fidelity_to_worst}]}\label{fig:noise_reduction}
\end{figure}

\subsection{Numerical simulations}

Tailoring experimental noise into stochastic noise via our technique provides 
several dramatic advantages, which we now illustrate via numerical simulations.
Our simulations are all of six-qubit circuits with the canonical division into 
easy and hard gates. That is, the easy gates are composed of Pauli gates and 
the phase gate $R = |0\>\< 0| + i|1\>\< 1|$, while the hard gates are the 
Hardamard, $\pi/8$ gate $T = \sqrt{R}$ and the two-qubit controlled-Z gate 
$\Delta(Z) = |0\>\<0|\otimes I + |0\>\<0|\otimes Z$. Such circuits are 
universal for quantum computation and naturally suit many fault-tolerant 
settings, including CSS codes with a transversal $T$ gate (such 
as the 15-qubit Reed-Muller code), color codes and the surface code. 

We quantify the total noise in a noisy implementation $\mc{C}_{\rm noisy}$ of 
an ideal circuit $\mc{C}_{\rm id}$ by the variational distance
\begin{align}\label{eq:tvd}
\tau_{\rm noisy} = \sum_j \tfrac{1}{2}|{\rm Pr}(j|\mc{C}_{\rm noisy}) - {\rm 
Pr}(j|\mc{C}_{\rm id})|
\end{align}
between the probability distributions for ideal computational basis 
measurements after applying $\mc{C}_{\rm noisy}$ and $\mc{C}_{\rm id}$ to a 
system initialized in the $\ket{0}\tn{n}$ state. We do not maximize over states 
and measurements, rather, our results indicate the effect of noise under 
practical choices of preparations and measurements. 

For our numerical simulations, we add gate-dependent over-rotations to each 
gate, that is, we perturb one of the eigenvectors of each gate $U$ by 
$e^{i\delta_U}$. For single-qubit gates, the choice of eigenvector is 
irrelevant (up to a global phase), while for the two-qubit $\Delta(Z)$ gate, we 
add the phase to the $|11\>$ state. 

We perform two sets of numerical simulations to illustrate two particular 
properties. First, \cref{fig:vary_noise} shows that our technique introduces a 
larger relative improvement as the infidelity decreases, that is, approximately 
a factor of two on a log scale, directly analogous to the $r/\sqrt{r}$ scaling 
for the worst case error (although recall that our simulations are for 
computational basis states and measurements and do not maximize the error over 
preparations and measurements). For these simulations, we set $\delta_U$ so 
that the $\Delta(Z)$ gate has an infidelity of $r[\Delta(Z)]$ and so that all 
single-qubit gates have an infidelity of $r[\Delta(Z)]/10$ (regardless of 
whether they are included in the ``easy'' or the ``hard'' set). For the bare 
circuits (blue circles), each data point is the variational distance of a 
randomly-chosen six-qubit circuit with a hundred alternating rounds of easy and 
hard gates, each sampled uniformly from the sets of all possible easy and hard 
gate rounds respectively. For the tailored circuits (red squares), each data 
point is the variational distance between ${\rm Pr}(j|\rm{C}_{\rm id})$ and the 
probability distribution ${\rm Pr}(j|\rm{C}_{\rm noisy})$ averaged over $10^3$ 
randomizations of the bare circuit obtained by replacing the easy gates 
with (compiled) dressed gates as in \cref{eq:dressed}.

Second, \cref{fig:vary_gates} shows that the typical error for both the bare 
and tailored circuits grows approximately linearly with the length of the 
circuit. This suggests that, for typical circuits, the primary reason that the 
total error is reduced by our technique is not because it prevents the 
worst-case quadratic accumulation of fidelity with the circuit length (although 
it does achieve this). Rather, the total error is reduced because the 
contribution from each error location is reduced, where the number of error 
locations grows linearly with the circuit size. For these simulations, we set 
$\delta_U$ so that the $\Delta(Z)$ gate has an infidelity of $10^{-3}$ and the 
easy gates have infidelities of $10^{-5}$. For the bare circuits (blue 
circles), each data point is the variational distance of a randomly-chosen 
six-qubit circuit as above with $K$ alternating rounds of easy and hard gates, 
where $K$ varies from five to a hundred. The tailored circuits (red squares) 
again give the variational distance between the ideal distribution and the 
probability distribution averaged over $10^3$ randomizations of the bare 
circuit.

\begin{figure}
\includegraphics[height=6cm]{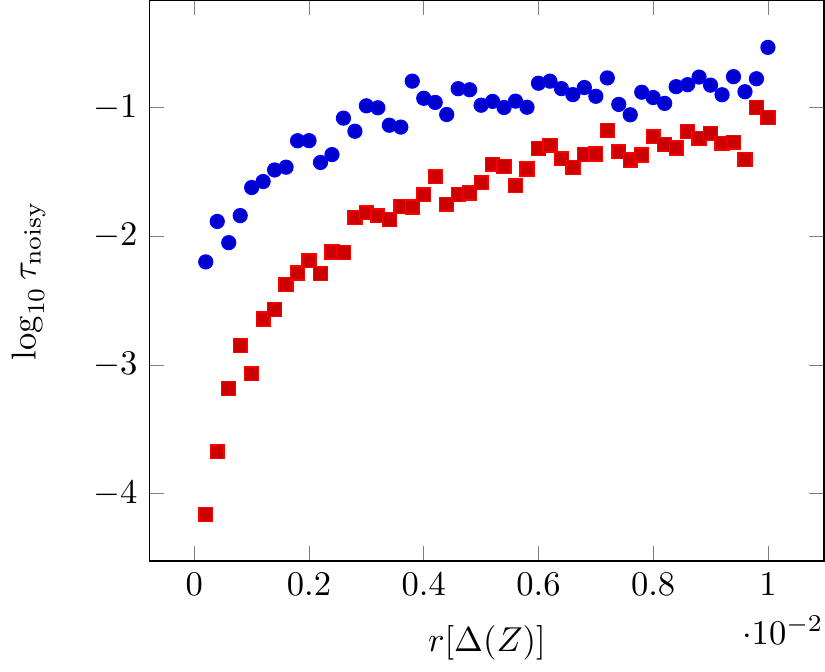}	
\caption{Semilog plots of the error $\tau_{\rm noise}$ from 
	\cref{eq:tvd} with respect to computational basis measurement 
	outcomes as a function of the average gate error $r[\Delta(Z)]$ of the 
	noise on the $\Delta(Z)$ gate in six-qubit bare (blue circles) and 
	tailored (red squares) circuits. Each data point corresponds to an 
	independent random circuit with 100 cycles, where the (gate-dependent) 
	noise on each gates is an over-rotation about the relevant eigenbasis with 
	infidelity $r[\Delta(Z)]$ for the $\Delta(Z)$ gates and $r[\Delta(Z)]/10$ 
	for all single-qubit gates. The data points for the tailored noise 
	correspond to an average over $10^3$ independent randomizations of the 
	corresponding bare circuit via \cref{eq:dressed}. The total error for the 
	bare and tailored circuits differs by a factor of approximately 2 on a log 
	scale, mirroring the separation between the worst-case errors for 
	stochastic and unitary channels from \cref{eq:fidelity_to_worst} (although 
	here the error is not maximized over preparations and 
	measurements).}\label{fig:vary_noise}
\end{figure}

\begin{figure}
\includegraphics[height=6cm]{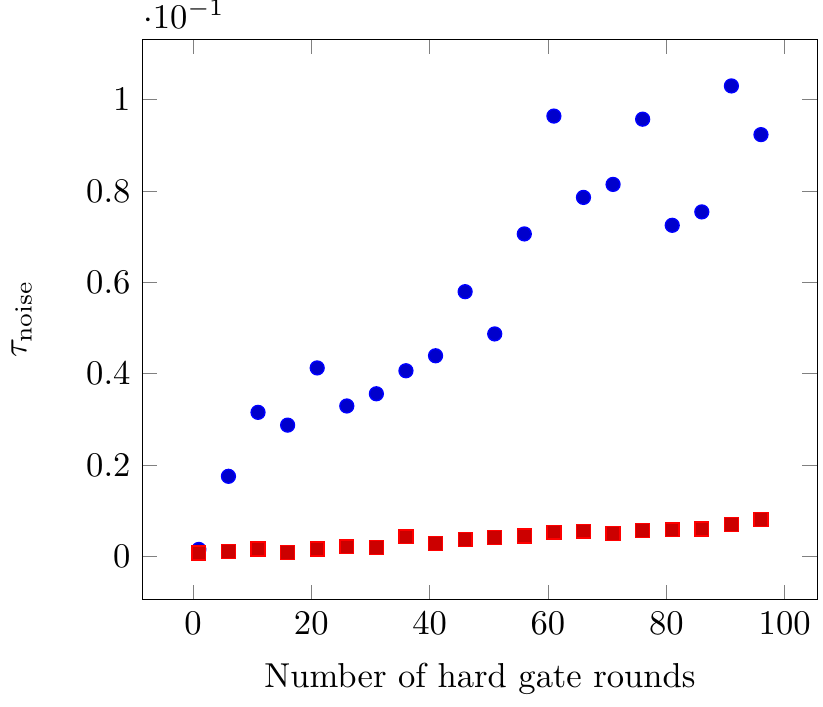}	
\caption{Plots of the error $\tau_{\rm noise}$ from 
	\cref{eq:tvd} with respect to computational basis measurement 
	outcomes as a function of the average gate error $r[\Delta(Z)]$ of the 
	noise on the $\Delta(Z)$ gate in six-qubit bare (blue circles) and 
	tailored (red squares) circuits. Each data point corresponds to an 
	independent random circuit with $K$ cycles, where the (gate-dependent) 
	noise on each gates is an over-rotation about the relevant eigenbasis with 
	infidelity $10^{-3}$ for the $\Delta(Z)$ gates and $10^{-5}$ for all 
	single-qubit gates. The data points for the tailored noise correspond to an 
	average over $10^3$ independent randomizations of the corresponding bare 
	circuit via \cref{eq:dressed}. The error rate grows approximately linearly 
	with the number of gate cycles, suggesting that the dominant reason for the 
	error suppression is that the error at each location is suppressed (where 
	there are a linear number of total locations), rather than the suppression 
	of possible quadratic accumulation of coherent errors 
	between locations.}\label{fig:vary_gates}
\end{figure}

\section{Discussion}

We have shown that arbitrary Markovian noise processes can be reduced to 
effective Pauli processes by compiling different sets of uniformly random gates 
into sequential operations. This randomized compiling technique can reduce the 
worst-case error rate by orders of magnitude and enables threshold estimates 
for general noise models to be obtained directly from threshold estimates for 
Pauli noise. Physical implementations can then be evaluated by directly 
comparing these threshold estimates to the average error rate $r$ estimated via 
efficient experimenal techniques, such as randomized benchmarking, to determine 
whether the experimental implementation has reached the fault-tolerant regime. 
More specifically, the average error rate $r$ is that of the tailored channel 
for the composite noise on a round of easy and hard gates and this can be 
directly estimated using interleaved randomized benchmarking with the relevant 
choice of group~\cite{Magesan2012}.

Our technique can be applied directly to gate sets that are universal for 
quantum computation, including all elements in a large class of fault-tolerant 
proposals. Moreover, our technique only requires \textit{local} gates to 
tailor general noise on multi-qubit gates into Pauli noise. Our numerical 
simulations in \cref{fig:vary_noise,fig:vary_gates} demonstrate that our 
technique can reduce worst-case errors by orders of magnitude. Furthermore, our 
scheme will generally 
produce an even greater effect as fault-tolerant protocols are scaled up, since 
fault-tolerant protocols are designed to suppress errors, for 
example, $\epsilon\to \epsilon^k$ for some scale factor $k$ (e.g., number of 
levels of concatenation), so any reduction at the physical level is improved exponentially with $k$.

A particularly significant open problem is the robustness of our technique to 
noise that remains non-Markovian on a time-scale longer than a typical gate 
time. Non-Markovian noise can be mitigated by techniques such as 
randomized dynamic decoupling~\cite{Viola2005,Santos2006}, which correspond to 
applying random sequences of Pauli operators to echo out non-Markovian 
contributions. Due to the random gates compiled in at each time step, we expect 
that our technique may also suppress non-Markovian noise in a similar manner.

\section{Methods}

We now prove \cref{thm:gate_dependent,thm:local}.

\begin{thm}\label{thm:gate_dependent_general}
Let $\mc{C}_{\rm GD}$ and $\mc{C}_{\rm GI}$ be tailored circuits with 
gate-dependent and gate-independent noise on the easy gates respectively. Then
\begin{align}
\|\mc{C}_{\rm GD}-\mc{C}_{\rm GI}\|_\diamond 
\leq \sum_{k=1}^K \md{E}_{\vec{T}_1,\ldots,\vec{T}_K}  
\|\mc{E}(\vec{\tilde{C}}_k)-\mc{E}_k^{\mbf{T}}\|_\diamond.
\end{align}
When $\mbf{T}$ is a group normalized by $\mbf{C}$, this can be improved to
\begin{align}
\|\mc{C}_{\rm GD},\mc{C}_{\rm GI}\|_\diamond 
&\leq \sum_{k=2}^K 2\md{E}_{\vec{\tilde{C}}_k}\| \mc{E}(\vec{\tilde{C}}_k) - 
\mc{E}_k^{\mbf{T}}\|\epsilon\left[\mc{E}(G_{k-1})\mc{E}_{k-1}^{\mbf{T}}\right]  
 \notag\\
&\quad + \md{E}_{\vec{\tilde{C}}_1}\|\mc{E}(\vec{\tilde{C}}_1) - 
\mc{E}_1^{\mbf{T}}\|_\diamond.
\end{align}
\end{thm}

\begin{proof}
Let
\begin{align}
\mc{A}_k &=G_k \mc{E}(G_k)\mc{E}(\vec{\tilde{C}}_k)\vec{\tilde{C}}_k
\notag\\
\mc{B}_k &= G_k \mc{E}(G_k)\mc{E}_k^{\mbf{T}} \vec{\tilde{C}}_k,
\end{align}
where $\mc{A}_k$ and $\mc{B}_k$ implicitly depend on the choice of twirling 
gates. Then the tailored circuits under gate-dependent and gate-independent 
noise are
\begin{align}
\mc{C}_{\rm GD} &= \md{E}_a\mc{A}_{K:1}
\notag\\
\mc{C}_{\rm GI} &= \md{E}_a\mc{B}_{K:1},
\end{align} 
respectively, where $\mc{X}_{a:b} = \mc{X}_a \ldots \mc{X}_b$ (note that this 
product is non-commutative) with $\mc{A}_{K:K+1} = \mc{B}_{0:1} = \mc{I}$ and 
the expectation is over all $\vec{T}_a$ for $a=1,\ldots,K$. Then by a 
straightforward induction argument,
\begin{align}\label{eq:telescope}
\mc{A}_{K:1} - \mc{B}_{K:1}
&= \sum_{k=1}^K \mc{A}_{K:k+1} (\mc{A}_k - 
\mc{B}_k) \mc{B}_{k-1:1} 
\end{align}
for any fixed choice of the twirling gates. By the triangle inequality,
\begin{align}\label{eq:simple}
\|\mc{C}_{\rm GD} - \mc{C}_{\rm GI} \|_\diamond
&= \|\md{E}_a \sum_{k=1}^K \mc{A}_{K:k+1} (\mc{A}_k - \mc{B}_k) \mc{B}_{k-1:1}  
\|_\diamond \notag\\
&\leq \md{E}_a \sum_{k=1}^K \| \mc{A}_{K:k+1} (\mc{A}_k - \mc{B}_k) 
\mc{B}_{k-1:1}  \|_\diamond \notag\\
&\leq \md{E}_a \sum_{k=1}^K 
\|\mc{E}(\vec{\tilde{C}}_k)-\mc{E}_k^{\mbf{T}}\|_\diamond,
\end{align}
where the second inequality follows from the submultiplicativity
\begin{align}
\|\mc{AB}\|_\diamond &\leq \|\mc{A}\|_\diamond \|\mc{B}\|_\diamond
\end{align}
of the diamond norm and the normalization $\|\mc{A}\|_\diamond =1$ which holds 
for all quantum channels $\mc{A}$.

We can substantially improve the above bound by evaluating some of the averages 
over twirling gates before applying the triangle inequality. In particular, 
leaving the average over $\vec{T}_{k-1}$ inside the diamond norm in 
\cref{eq:simple} for every term except $k=1$ gives
\begin{align}
\|\mc{C}_{\rm GD} - \mc{C}_{\rm GI} \|_\diamond
&\leq  \sum_{k=1}^K \md{E}_{a\neq k-1} \| \md{E}_{k-1} \delta_k\gamma_k 
\|_\diamond \notag\\
&\quad + \md{E}_a \|\mc{E}(\vec{\tilde{C}}_1)-\mc{E}_1^{\mbf{T}}\|_\diamond,
\end{align}
where
\begin{align}
\delta_k &= \mc{E}(\vec{\tilde{C}}_k) - \mc{E}_k^{\mbf{T}} \notag\\
\gamma_k &= \vec{\tilde{C}}_k G_{k-1} \mc{E}(G_{k-1})\mc{E}_{k-1}^{\mbf{T}} 
\vec{T}_{k-1},
\end{align}
and $\delta_k\gamma_k$ is the only factor of $\mc{A}_{K:k+1} (\mc{A}_k - 
\mc{B}_k) \mc{B}_{k-1:1}$ that depends on $\vec{T}_{k-1}$. Substituting  
$\mc{E}(G_{k-1})\mc{E}_{k-1}^{\mbf{T}} = (\mc{E}(G_{k-1})\mc{E}_{k-1}^{\mbf{T}} 
- \mc{I}) + \mc{I}$ in $\gamma_k$ gives
\begin{align}
\md{E}_{k-1} \delta_k
&= \md{E}_{k-1} \delta_k \vec{\tilde{C}}_k G_{k-1}  
\left[\mc{E}(G_{k-1})\mc{E}_{k-1}^{\mbf{T}} - \mc{I}\right] \vec{T}_{k-1} 
\notag\\
&\quad + \md{E}_{k-1} \delta_k \vec{T}_k\vec{C}_k G_{k-1} ,
\end{align}
where the only factor in the second term that depends on $\vec{T}_{k-1}$ is 
$\delta_k$. When $\mbf{T}$ is a group normalized by $\mbf{C}$,
\begin{align}
\md{E}_{k-1} \delta_k &= \md{E}_{k-1}
\mc{E}(\vec{\tilde{C}}_k) - \mc{E}_k^{\mbf{T}} \notag\\
&= \md{E}_{k-1}
\mc{E}(\vec{C}_k[\vec{C}_k\ct\vec{T}_k\vec{C}_k]\vec{T}_{k-1}^c) - 
\mc{E}_k^{\mbf{T}} \notag\\
&= \md{E}_{\vec{T}'} 
\mc{E}(\vec{C}_k\vec{T}') - \mc{E}_k^{\mbf{T}} \notag\\
&= 0
\end{align}
for any fixed value of $\vec{T}_k$, using the fact that $\{hg:g\in\mbf{G}\} = 
\mbf{G}$ for any group $\mbf{G}$ and $h\in\mbf{G}$ and that $\md{E}_{\vec{T}'} 
\mc{E}(\vec{C}_k\vec{T}')$ is independent of $\vec{T}_k$. Therefore
\begin{align}\label{eq:guts}
\|\mc{C}_{\rm GD} - \mc{C}_{\rm GI}\|_\diamond
&= \|\md{E}_j (\mc{A}_{K:1} - \mc{B}_{K:1}) \|_\diamond \notag\\
&\leq \|\delta_1\|_\diamond + \sum_{k=2}^K \md{E}_{j\neq k-1}\| 
\md{E}_{k-1}\delta_k \gamma_k  \|_\diamond \notag\\
&\leq \|\delta_1\|_\diamond + \sum_{k=2}^K \md{E}_{j\neq k-1}\| 
\md{E}_{\vec{T}_{k-1}}\delta_k \vec{\tilde{C}}_k  \notag\\
&\quad\times  G_{k-1}\left[\mc{E}(G_{k-1})\mc{E}_{k-1}^{\mbf{T}} - 
\mc{I}\right] 
\vec{T}_{k-1}\|_\diamond \notag\\
&\leq \sum_{k=2}^K \md{E}_j \| \delta_k\|_\diamond
\|\mc{E}(G_{k-1})\mc{E}_{k-1}^{\mbf{T}} - \mc{I} \|_\diamond \notag\\
&\quad + \|\delta_1\|_\diamond,
\end{align}
where we have had to split the sum over $k$ as $\vec{T}_0$ is fixed to the 
identity.
\end{proof}

\begin{thm}
For arbitrary noise,
\begin{align}
\md{E}_{\vec{\tilde{C}}_k}\| \mc{E}(\vec{\tilde{C}}_k) - 
\mc{E}_k^{\mbf{T}} \|_\diamond 
\leq 2\epsilon(\mc{E}_k^{\mbf{T}}) + 2\sqrt{\md{E}_{\vec{\tilde{C}}_k} 
	\epsilon[\mc{E}(\vec{\tilde{C}}_k)]^2}.
\end{align}	
For $n$-qubit circuits with local noise on the easy gates, 
\begin{align}
\md{E}_{\vec{\tilde{C}}_k}\| \mc{E}(\vec{\tilde{C}}_k) - 
\mc{E}_k^{\mbf{T}} \|_\diamond &\leq 
\sum_{j=1}^n 4\sqrt{6r\Bigl[\overline{\mc{E}}_j^{\mbox{T}}(C_{j,k})\Bigr]}
\end{align}
for $k=1,\ldots,K$, where $\mc{E}_{j.k}^{\mbox{T}} 
=\md{E}_{\tilde{C}_{j,k}}\mc{E}_j (\tilde{C}_{j,k})$ is the local noise on the 
$j$th qubit averaged over the dressed gates in the $k$th cycle.
\end{thm}

\begin{proof}
First note that, by the triangle inequality,
\begin{align}\label{eq:bad_bound}
\md{E}_{\vec{\tilde{C}}_k}\| \mc{E}(\vec{\tilde{C}}_k) - 
\mc{E}_k^{\mbf{T}} \|_\diamond 
&= \md{E}_{\vec{\tilde{C}}_k}\| \mc{E}(\vec{\tilde{C}}_k) - \mc{I} 
+ \mc{I} - \mc{E}_k^{\mbf{T}} \|_\diamond  \notag \\
&\leq \md{E}_{\vec{\tilde{C}}_k}\| \mc{E}(\vec{\tilde{C}}_k) - \mc{I} 
\|_\diamond + \|\mc{I} - \mc{E}_k^{\mbf{T}} \|_\diamond \notag \\
&\leq \md{E}_{\vec{\tilde{C}}_k}2\epsilon[\mc{E}(\vec{\tilde{C}}_k)] + 
2\epsilon(\mc{E}_k^{\mbf{T}}).
\end{align}
By the Cauchy-Schwarz inequality,
\begin{align}\label{eq:CS_expectation}
\left(\md{E}_{\vec{\tilde{C}}_k}\epsilon[\mc{E}(\vec{\tilde{C}}_k)] \right)^2
&= \left(\sum_{\vec{\tilde{C}}_k} |\# \vec{\tilde{C}}_k|^{-1} 
\epsilon[\mc{E}(\vec{\tilde{C}}_k)] \right)^2\notag\\
&\leq \left(\sum_{\vec{\tilde{C}}_k} |\# \vec{\tilde{C}}_k|^{-2} \right)
\left(\sum_{\vec{\tilde{C}}_k}\epsilon[\mc{E}(\vec{\tilde{C}}_k)]^2 
\right)\notag\\
&\leq |\# \vec{\tilde{C}}_k|^{-1}
\left(\sum_{\vec{\tilde{C}}_k}\epsilon[\mc{E}(\vec{\tilde{C}}_k)]^2 
\right)\notag\\
&= \md{E}_{\vec{\tilde{C}}_k} \epsilon[\mc{E}(\vec{\tilde{C}}_k)]^2 ,
\end{align}	
where $\#\vec{\tilde{C}}_k$ is the number of different dressed gates in the 
$k$th round.

For local noise, that is, noise of the form $\mc{E}_1\otimes \ldots \mc{E}_n$ 
where $\mc{E}_j$ is the noise on the $j$th qubit, 
\begin{align}
\epsilon[\mc{E}(\vec{\tilde{C}}_k)]&= \tfrac{1}{2}
\| \bigotimes_{j=1}^n\mc{E}_j(\tilde{C}_{j,k}) - \mc{I} \|_\diamond \notag\\
&\leq  \sum_{j=1}^n \tfrac{1}{2}
\| \mc{E}_j(\tilde{C}_{j,k}) - \mc{I} \|_\diamond \notag\\
&\leq \sum_{j=1}^n \epsilon[\mc{E}(\tilde{C}_{j,k})],
\end{align}
where we have used the analog of \cref{eq:telescope} for the tensor product and
\begin{align}\label{eq:diamond_tensor}
\|A\otimes B\|_\diamond &\leq \|A\|_\diamond\|B\|_\diamond,
\end{align}
which holds for all $A$ and $B$ due to the submultiplicativity of the diamond 
norm, and the equality $\|A\otimes I\|_\diamond = \|A\|_\diamond$. Similarly,
\begin{align}
\epsilon(\mc{E}_k^{\mbf{T}}) = \sum_{j=1}^n \epsilon(\mc{E}_{j,k}^{\mbf{T}})
\end{align}
where $\mc{E}_{j,k}^{\mbf{T}} = \md{E}_{T_{j,k-1},T_{j,k}} 
\mc{E}(\tilde{C}_{j,k})$. We then have
\begin{align}
\md{E}_{\vec{\tilde{C}}_k}\epsilon[\mc{E}(\vec{\tilde{C}}_k)]
&\leq \sum_{j=1}^n \md{E}_{\tilde{C}_{j,k}}\epsilon[\mc{E}(\tilde{C}_{j,k})] 
\notag\\
&\leq \sum_{j=1}^n 
\sqrt{\md{E}_{\tilde{C}_{j,k}}\epsilon[\mc{E}(\tilde{C}_{j,k})]^2}
\end{align}
where the second inequality is due to the Cauchy-Schwarz inequality as in 
\cref{eq:CS_expectation}. Returning to \cref{eq:bad_bound}, we have
\begin{align}
\md{E}_{\vec{\tilde{C}}_k}\| \mc{E}(\vec{\tilde{C}}_k) - 
\mc{E}_k^{\mbf{T}} \|_\diamond 
&\leq \sum_{j=1}^n 2\epsilon(\mc{E}_{j,k}^{\mbf{T}}) + 
2\sqrt{\md{E}_{\tilde{C}_{j,k}}\epsilon[\mc{E}(\tilde{C}_{j,k})]^2} \notag\\
&\leq \sum_{j=1}^n 2\sqrt{6r(\mc{E}_{j,k}^{\mbf{T}})} + 
2\sqrt{\md{E}_{\tilde{C}_{j,k}}6r[\mc{E}(\tilde{C}_{j,k})]} \notag\\
&\leq \sum_{j=1}^n 4\sqrt{6r(\mc{E}_{j,k}^{\mbf{T}})}
\end{align}
for local noise, where the second inequality follows from 
\cref{eq:fidelity_to_worst} with $d=2$ and the third from the linearity of the 
infidelity.
\end{proof}
\textit{Acknowledgments}---The authors acknowledge helpful discussions with 
Arnaud Carignan-Dugas, David Cory, Steve Flammia, Daniel Gottesman, Tomas 
Jochym-O'Connor and Raymond Laflamme. This research was supported by the U.S. 
Army Research Office through grant W911NF-14-1-0103, CIFAR, the Government of 
Ontario, and the Government of Canada through NSERC and Industry Canada.


\begin{thebibliography}{50}
\bibitem{Shor1999}  P. W. Shor, 
\href{http://dx.doi.org/10.1137/S0036144598347011}{SIAM Rev. \textbf{41}, 303 
(1999)}.
\bibitem{Lloyd1996} S. Lloyd, 
\href{http://dx.doi.org/10.1126/science.273.5278.1073}{Science \textbf{273}, 
1073 (1996)}.
\bibitem{Shor1995} P. W. Shor, 
\href{http://dx.doi.org/10.1103/PhysRevA.52.R2493}{Phys. Rev. A \textbf{52}, 
2493 (1995)}.
\bibitem{Gottesman1996} D. Gottesman, 
\href{http://dx.doi.org/10.1103/PhysRevA.54.1862}{Phys. Rev. A \textbf{54}, 
1862 (1996)}.
\bibitem{Knill1998} E. Knill, 
\href{http://dx.doi.org/10.1126/science.279.5349.342}{Science \textbf{279}, 342 
(1998)}.
\bibitem{Aharonov1999} D. Aharonov and M. Ben-Or, 
\href{http://dx.doi.org/10.1137/S0097539799359385}{SIAM J. Comput. \textbf{38}, 
1207 (1999)}.
\bibitem{Calderbank2002} A. R. Calderbank, E. M. Rains, P. W. Shor, and
N. J. A. Sloane, \href{http://dx.doi.org/10.1109/18.681315}{IEEE Trans. Inf. 
Theory \textbf{44}, 1369 (1998)}.
\bibitem{Kitaev2003} A. Kitaev, 
\href{http://dx.doi.org/10.1016/S0003-4916(02)00018-0}{Ann. Phys. (N. Y). 
\textbf{303}, 2 (2003)}.

\bibitem{Aliferis2007a} P. Aliferis and A. W. Cross, 
\href{http://dx.doi.org/10.1103/PhysRevLett.98.220502}{Phys. Rev. Lett. 
\textbf{98}, 220502 (2007)}.

\bibitem{Aliferis2008} P. Aliferis, D. Gottesman, and J. Preskill, 
\href{http://www.rintonpress.com/xxqic8/qic-8-34/0181-0244.pdf}{Quantum Inf. \& 
Comput. \textbf{8}, 181 (2008)}.

\bibitem{Knill2005}  E. Knill, 
\href{http://dx.doi.org/10.1038/nature03350}{Nature \textbf{434}, 39 (2005)}.

\bibitem{Wang2011}  D. S. Wang, A. G. Fowler, and L. C. L. Hollenberg, 
\href{http://dx.doi.org/10.1103/PhysRevA.83.020302}{Phys. Rev. A \textbf{83}, 
020302 (2011)}.

\bibitem{Duclos-Cianci2010} G. Duclos-Cianci and D. Poulin, 
\href{http://dx.doi.org/10.1103/PhysRevLett.104.050504}{Phys. Rev. Lett. 
\textbf{104}, 050504 (2010)}.

\bibitem{Wootton2012} J. R. Wootton and D. Loss, 
\href{http://dx.doi.org/10.1103/PhysRevLett.109.160503}{Phys. Rev. Lett. 
\textbf{109}, 16053 (2012)}.

\bibitem{Bombin2012} H. Bombin, R. S. Andrist, M. Ohzeki, H. G. Katzgraber,
and M. A. Martin-Delgado, 
\href{http://dx.doi.org/10.1103/PhysRevX.2.021004}{Phys. Rev. X \textbf{2}, 
021004 (2012)}.

\bibitem{Puzzuoli2014} D. Puzzuoli, C. Granade, H. Haas, B. Criger, E. Magesan, 
and D. G. Cory, \href{http://dx.doi.org/10.1103/PhysRevA.89.022306}{Phys. Rev. 
A \textbf{89}, 022306 (2014)}.

\bibitem{Kitaev1997} A. Kitaev, 
\href{http://dx.doi.org/10.1070/RM1997v052n06ABEH002155}{Russ. Math. Surv. 
\textbf{52}, 1191 (1997)}.

\bibitem{Emerson2005} J. Emerson, R. Alicki, and K. \.{Z}yczkowski, 
\href{http://dx.doi.org/10.1088/1464-4266/7/10/021}{J. Opt. B 
\textbf{7}, S347 (2005)}.

\bibitem{Emerson2007} J. Emerson, M. Silva, O. Moussa, C. A. Ryan, M. Laforest, 
J. Baugh, D. G. Cory, and R. Laflamme, 
\href{http://dx.doi.org/10.1126/science.1145699}{Science \textbf{317}, 1893 
	(2007)}.

\bibitem{Knill2008} E. Knill, D. Leibfried, R. Reichle, J. Britton, R. B.
Blakestad, J. D. Jost, C. Langer, R. Ozeri, S. Seidelin,
and D. J. Wineland, \href{http://dx.doi.org/ 10.1103/PhysRevA.77.012307}{Phys. 
Rev. A 77, 012307 (2008)}.

\bibitem{Dankert2009} C. Dankert, R. Cleve, J. Emerson, and E. Livine, 
\href{http://dx.doi.org/10.1103/PhysRevA.80.012304}{Phys. Rev. A \textbf{80}, 
	012304 (2009)}.

\bibitem{Magesan2011} E. Magesan, J. M. Gambetta, and J. Emerson, 
\href{http://dx.doi.org/10.1103/PhysRevLett.106.180504}{Phys.
	Rev. Lett. \textbf{106}, 180504 (2011)}.

\bibitem{Beigi2011} S. Beigi and R. König, 
\href{http://dx.doi.org/10.1088/1367-2630/13/9/093036}{New J. Phys. 
\textbf{13}, 093036 (2011)}.

\bibitem{Wallman2014} J. J. Wallman and S. T. Flammia, 
\href{http://dx.doi.org/10.1088/1367-2630/16/10/103032}{New J. Phys. 
\textbf{16}, 103032 (2014)}.

\bibitem{Magesan2012a} E. Magesan, J. M. Gambetta, and J. Emerson, 
\href{http://dx.doi.org/10.1103/PhysRevA.85.042311}{Phys. Rev. A \textbf{85}, 
042311 (2012)}.

\bibitem{Sanders2015} Y. R. Sanders, J. J. Wallman, and B. C. Sanders, 
\href{http://dx.doi.org/10.1088/1367-2630/18/1/012002}{New J. Phys. 
	\textbf{18}, 012002 (2015)}.

\bibitem{Knill2004a} E. Knill, \href{http://arxiv.org/abs/quant-ph/0404104}{	
arXiv:quant-ph/0404104}.

\bibitem{Kern2005} O. Kern, G. Alber, and D. L. Shepelyansky, 
\href{http://dx.doi.org/10.1140/epjd/e2004-00196-9}{Eur. Phys.
J. D \textbf{32}, 153 (2005)}.

\bibitem{Aaronson2004} S. Aaronson, and D. Gottesman
\href{http://dx.doi.org/10.1103/PhysRevA.70.052328}{Phys. Rev. A \textbf{70}, 
052328 (2004)}.

\bibitem{Eastin2009} B. Eastin and E. Knill, 
\href{http://dx.doi.org/10.1103/PhysRevLett.102.110502}{Phys. Rev. Lett. 
	\textbf{102}, 110502 (2009)}.

\bibitem{Beverland2014} M. E. Beverland, O. Buerschaper, R. Koenig,
F. Pastawski, J. Preskill, and S. Sijher, 
\href{http://dx.doi.org/10.1063/1.4939783}{J. Math. Phys. \textbf{57}, 022201 
(2016)}.

\bibitem{Bravyi2005} S. Bravyi and A. Kitaev, 
\href{http://dx.doi.org/10.1103/PhysRevA.71.022316}{Phys. Rev. A \textbf{71}, 
022316 (2005)}.

\bibitem{Bombin2015} H. Bomb{\'i}n, 
\href{http://dx.doi.org/10.1088/1367-2630/17/8/083002}{New J. Phys. 
\textbf{17}, 083002 (2015)}.

\bibitem{Wallman2015} J. J. Wallman, C. Granade, R. Harper, and S. T. Flammia, 
\href{http://dx.doi.org/10.1088/1367-2630/17/11/113020}{New J. Phys. 
\textbf{17}, 113020 (2015)}.

\bibitem{Kueng2015} R. Kueng, D. M. Long, A. C. Doherty, and S. T. Flammia, 
\href{http://arxiv.org/abs/1510.05653}{arXiv:1510.05653 [quant-ph]}.

\bibitem{Wallman2015b} J. J. Wallman, 
\href{http://arxiv.org/abs/1511.00727}{arXiv:1511.00727 [quant-ph]}.

\bibitem{Magesan2012}  E. Magesan, J. M. Gambetta, B. R. Johnson, C. A. Ryan,
J. M. Chow, S. T. Merkel, M. P. da Silva, G. A. Keefe, M. B. Rothwell, T. A. 
Ohki, M. B. Ketchen, and M. Steffen, \href{http://dx.doi.org/ 
10.1103/PhysRevLett.109.080505}{Phys. Rev. Lett. \textbf{109}, 080505 (2012)}.

\bibitem{Viola2005} L. Viola and E. Knill, 
\href{http://dx.doi.org/10.1103/PhysRevLett.94.060502}{Phys. Rev. Lett. 
\textbf{94}, 060502 (2005)}.

\bibitem{Santos2006}  L. F. Santos and L. Viola, 
\href{http://dx.doi.org/10.1103/PhysRevLett.97.150501}{Phys. Rev. Lett. 
\textbf{97}, 150501 (2006)}.
\end{thebibliography}
\end{document}